\newcommand{\arxivversion}[1]{#1}
\newcommand{\aamasversion}[1]{}

\arxivversion{\documentclass{article}}
\aamasversion{\documentclass[sigconf,anonymous]{aamas} }

\usepackage{graphicx} 
\usepackage{tcolorbox}
\usepackage{color}

\usepackage{xspace}
\usepackage{algorithm}
\usepackage{algpseudocode}
\usepackage{float}
\usepackage{amsmath,amsthm,amssymb,bbm,thm-restate,thmtools}
\usepackage{todonotes}

\arxivversion{\usepackage{charter,eulervm}
\usepackage{fullpage}
}

\newcommand{\NPC}{\ensuremath{\mathsf{NP}}\text{-complete}\xspace}
\newcommand{\NPH}{\ensuremath{\mathsf{NP}}\text{-hard}\xspace}
\newcommand{\PNPH}{para-\ensuremath{\mathsf{NP}\text{-hard}}\xspace}
\newcommand{\el}{\ensuremath{\ell}\xspace}

\newcommand{\WOH}{\ensuremath{\mathsf{W[1]}}-hard\xspace}

\newcommand{\FPT}{\ensuremath{\mathsf{FPT}}\xspace}

\let\oldlambda\lambda
\renewcommand{\lambda}{\ensuremath{\oldlambda}\xspace}
\let\oldalpha\alpha
\renewcommand{\alpha}{\ensuremath{\oldalpha}\xspace}
\let\oldDelta\Delta
\renewcommand{\Delta}{\ensuremath{\oldDelta}\xspace}

\newcommand{\yes}{{\sc yes}\xspace}

\newcommand{\caveat}{\ensuremath{\mathsf{CoNP\subseteq NP/Poly}}\xspace}

\renewcommand{\AA}{\ensuremath{\mathcal A}\xspace}
\newcommand{\BB}{\ensuremath{\mathcal B}\xspace}
\newcommand{\CC}{\ensuremath{\mathcal C}\xspace}

\newcommand{\EE}{\ensuremath{\mathcal E}\xspace}

\newcommand{\GG}{\ensuremath{\mathcal G}\xspace}

\newcommand{\OO}{\ensuremath{\mathcal O}\xspace}
\newcommand{\PP}{\ensuremath{\mathcal P}\xspace}

\newcommand{\UU}{\ensuremath{\mathcal U}\xspace}
\newcommand{\VV}{\ensuremath{\mathcal V}\xspace}

\newcommand{\KE}{{\sc Kidney Exchange}\xspace}

\newcommand{\BP}{{\sc Bin Packing}\xspace}

\newcommand{\tw}{\ensuremath{\tau}\xspace}
\let\mydelta\delta
\renewcommand{\delta}{\ensuremath{\mydelta}\xspace}

\let\mytau\tau
\renewcommand{\tau}{\ensuremath{\mytau}\xspace}

\let\mytheta\theta
\renewcommand{\theta}{\ensuremath{\mytheta}\xspace}

\let\mygamma\gamma
\renewcommand{\gamma}{\ensuremath{\mygamma}\xspace}

\let\myGamma\Gamma
\renewcommand{\Gamma}{\ensuremath{\myGamma}\xspace}
\usepackage{tikz}
\usetikzlibrary{arrows,positioning}

\arxivversion{
\newtheorem{proposition}{\bf Proposition}

\newtheorem{theorem}{\bf Theorem}
\newtheorem{lemma}{\bf Lemma}

\newtheorem{definition}{\bf Definition}

}

\usepackage{cleveref}

\crefname{thm}{Theorem}{\bf Theorems}
\crefname{cor}{Corollary}{\bf Corollaries}

\crefname{theorem}{Theorem}{\bf Theorems}
\crefname{observation}{Observation}{\bf Observations}
\crefname{lemma}{Lemma}{\bf Lemmata}
\crefname{corollary}{Corollary}{\bf Corollaries}
\crefname{proposition}{Proposition}{\bf Propositions}
\crefname{definition}{Definition}{\bf Definitions}
\crefname{claim}{Claim}{\bf Claims}
\crefname{remark}{remark}{\bf Remark}
\crefname{reductionrule}{Reduction rule}{\bf Reduction rules}
\crefname{probdefinition}{Problem Definition}{Problem Definitions}

\usepackage{tikz}
\usetikzlibrary{positioning,arrows.meta}

\aamasversion{

\setcopyright{ifaamas}
\acmConference[AAMAS '26]{Proc.\@ of the 25th International Conference
on Autonomous Agents and Multiagent Systems (AAMAS 2026)}{May 25 -- 29, 2026}
{Paphos, Cyprus}{C.~Amato, L.~Dennis, V.~Mascardi, J.~Thangarajah (eds.)}
\copyrightyear{2026}
\acmYear{2026}
\acmDOI{}
\acmPrice{}
\acmISBN{}

\acmSubmissionID{1534}

}

\aamasversion{\title[Kidney Exchange]{Kidney Exchange: Faster Parameterized Algorithms and\\ Tighter Lower Bounds}}

\arxivversion{\title{Kidney Exchange: Faster Parameterized Algorithms and Tighter Lower Bounds}}

\arxivversion{

\usepackage{authblk}

\author[1]{Aritra Banik\thanks{aritra@niser.ac.in}}
\author[2]{Sujoy Bhore\thanks{sujoy@cse.iitb.ac.in}}
\author[3]{Palash Dey\thanks{palash.dey@cse.iitkgp.ac.in}}
\author[4]{Abhishek Sahu\thanks{abhiksheksahu@niser.ac.in}}

\affil[1,4]{National Institute of Science Education and Research, Bhubaneswar}
\affil[2]{Indian Institute of Technology Bombay}
\affil[3]{Indian Institute of Technology Kharagpur}
}

\aamasversion{
\begin{abstract}
    The kidney exchange mechanism allows many patient-donor pairs who are otherwise incompatible with each other to come together and exchange kidneys along a cycle. However, due to infrastructure and legal constraints, kidney exchange can only be performed in small cycles in practice. In reality, there are also some altruistic donors who do not have any paired patients. This allows us to also perform kidney exchange along paths that start from some altruistic donor. Unfortunately, the computational task is \NPC. To overcome this computational barrier, an important line of research focuses on designing faster algorithms, both exact and using the framework of parameterized complexity. 

    The standard parameter for the kidney exchange problem is the number $t$ of patients that receive a healthy kidney. The current fastest known deterministic \FPT algorithm for this problem, parameterized by $t$, is $\OO^\star\left(14^t\right)$. In this work, we improve this by presenting a deterministic \FPT algorithm that runs in time $\OO^\star\left((4e)^t\right)\approx \OO^\star\left(10.88^t\right)$. This problem is also known to be \WOH parameterized by the treewidth of the underlying undirected graph. A natural question here is whether the kidney exchange problem admits an \FPT algorithm parameterized by the pathwidth of the underlying undirected graph. We answer this negatively in this paper by proving that this problem is \WOH parameterized by the pathwidth of the underlying undirected graph. We also present some parameterized intractability results improving the current understanding of the problem under the framework of parameterized complexity.
\end{abstract}}

\aamasversion{
\keywords{kidney exchange, parameterized complexity, FPT algorithm, W-hardness, para-NP-hardness}}

\begin{document}


\aamasversion{
\pagestyle{fancy}
\fancyhead{}
}


\maketitle

\arxivversion{
\begin{abstract}
    The kidney exchange mechanism allows many patient-donor pairs who are otherwise incompatible with each other to come together and exchange kidneys along a cycle. However, due to infrastructure and legal constraints, kidney exchange can only be performed in small cycles in practice. In reality, there are also some altruistic donors who do not have any paired patients. This allows us to also perform kidney exchange along paths that start from some altruistic donor. Unfortunately, the computational task is \NPC. To overcome this computational barrier, an important line of research focuses on designing faster algorithms, both exact and using the framework of parameterized complexity. 

    The standard parameter for the kidney exchange problem is the number $t$ of patients that receive a healthy kidney. The current fastest known deterministic \FPT algorithm for this problem, parameterized by $t$, is $\OO^\star\left(14^t\right)$. In this work, we improve this by presenting a deterministic \FPT algorithm that runs in time $\OO^\star\left((4e)^t\right)\approx \OO^\star\left(10.88^t\right)$. This problem is also known to be \WOH parameterized by the treewidth of the underlying undirected graph. A natural question here is whether the kidney exchange problem admits an \FPT algorithm parameterized by the pathwidth of the underlying undirected graph. We answer this negatively in this paper by proving that this problem is \WOH parameterized by the pathwidth of the underlying undirected graph. We also present some parameterized intractability results improving the current understanding of the problem under the framework of parameterized complexity.
\end{abstract}}

\section{Introduction}
Kidney failure is a global health crisis, affecting millions worldwide. Patients suffering from acute renal failure typically have two main treatment options: dialysis or kidney transplantation. However, dialysis is associated with a significantly lower quality of life, and the average lifespan of recipients on dialysis is only around 10 years. As a result, kidney transplantation is the preferred choice for many recipients. Unfortunately, the gap between the demand for and supply of kidneys is vast. In the U.S., approximately 40,000 people are added to the kidney transplant waiting list every year, but only about 20,000 find a compatible donor~\cite{walsh2021}. This discrepancy has led to longer waiting times, with the current median wait for a kidney transplant from a deceased donor averaging 4.05 years~\cite{Stewart2023}. Tragically, thousands of patients are removed from the waiting list each year because they either die or become too ill to undergo the surgery.

In this challenging context, Kidney Paired Donation (KPD), also known as the Kidney Exchange Program, was introduced in 2000 to provide an innovative solution to the problem of incompatible donor-recipient pairs~\cite{rapaport1986case,ke04}. In KPD, recipients who have an incompatible donor can exchange their donor with another pair, forming a ``barter market"~\cite{ke04,ke07,AbrahamBS07} that increases the likelihood of successful kidney transplants. This system allows participants to find compatible donors more easily, often resulting in higher-quality matches that may lead to longer-lasting kidneys~\cite{segev2005kidney}. Since its introduction, KPD has expanded globally, with countries like the U.S., U.K., and India adopting the program, resulting in thousands of successful transplants and a significantly higher chance of matching better-quality kidneys~\cite{biro2021modelling}.

The central challenge in any KPD system is the Kidney Exchange Problem, how to allocate kidneys to maximize the number of transplants. This is typically modeled as a directed graph where each vertex represents a patient-donor pair. A directed edge from vertex u to v indicates that the donor from u is compatible with the recipient at v.

Kidney exchanges usually happen through cycles, where each recipient receives a kidney from the donor of the previous vertex. Since donors can withdraw once their paired recipient is treated, all surgeries in a cycle must be performed simultaneously to prevent defaults. Due to logistical constraints (each transplant involves two surgeries), only small cycles are allowed~\cite{ke05,AbrahamBS07,manlove2015paired,mak2017kidney,ashlagi2021kidney,freedman2020adapting}.

In addition to cycles, some programs allow chains initiated by altruistic donors (non-directed donors, or NDDs) who do not have a paired recipient. In a chain, a donor gives to a compatible recipient, whose paired donor donates to another recipient, and so on. Unlike cycles, chains can proceed non-simultaneously, since a broken chain does not leave any recipient worse off~\cite{anderson,dickerson2016position}. In the worst case, affected pairs can re-enter the exchange~\cite{anderson}.

With altruistic donors, the problem becomes finding a collection of disjoint cycles and chains that maximizes coverage of patient-donor pairs while complying with logistical constraints. This forms the core optimization problem in kidney exchanges.

Since the computational problem is \NPC, there is a long line of work focusing on designing heuristics that work reasonably well in practice. While existing clearing algorithms perform well in practice, understanding their theoretical limits provides insight into whether faster or more scalable mechanisms are even possible. Xiao and Wang~\cite{xiao2018exact} were the first to propose an exact algorithm for the kidney exchange problem that is guaranteed to output an optimal solution. Maiti and Dey~\cite{MaitiD22} extended this line of work by initiating the study of parameterized complexity of the kidney exchange problem. This work improves the current state of the art in this line of research.

\subsection{Our Contribution}
\begin{table}[h]
\renewcommand{\arraystretch}{1.5}
    \centering
    \begin{tabular}{c|c}
        Parameter & \KE \\\hline\hline
        $t$ & FPT ($\OO^\star\left(10.88^t\right)$) [\Cref{fpt}] \\\hline
        $t+\ell_p+\ell_c$  & No poly kernel [\Cref{cor:kidney-no-poly}] \\\hline
        pathwidth & \WOH$(\ell_c=0)$ [\Cref{thm:pathwidth-path}] \\\hline
        pathwidth & \WOH $(\ell_p=0)$ [\Cref{thm:pathwidth-cycle}] \\\hline
       DFVS+$\ell_p$+$\ell_c$   & Para-NP Hard [\Cref{cor-para}] \\\hline
    \end{tabular}
    \caption{Summary of results. $t:$ the number of patients receiving a kidney; $\el_p (\ell_c):$ maximum allowed length of a kidney exchange path (cycle); DFVS: size of any minimum directed feedback vertex set.}
    \label{tab:results}
\end{table}

In this work, we obtain the fastest known deterministic FPT algorithm for the \textsc{Kidney Exchange} problem with respect to the parameter~$t$, the number of patients helped.
The previous best algorithms of Johnson et al.~\cite{ijcai2024p9} ran in randomized time
$O^\star(4^t)$ and deterministic time $O^\star(14^t)$\footnote{$\mathcal{O}^\star(\cdot)$ notation hides polynomial factors in $n$.}.
Our approach improves the deterministic running time to
$O^\star((4e)^t) \approx O^\star(10.88^t)$.
The key ingredients are a careful combination of \emph{color coding}-
derandomized via $t$-perfect hash families, and a \emph{subset-based dynamic program} that systematically builds vertex-disjoint cycles and altruistic paths.
By adapting the fastest color-coding subroutines for detecting cycles and
(altruistic paths) of prescribed colors, we achieve the stated bound.

\begin{restatable*}{thm}{fptalgo}\label{fpt}
\KE parameterized by $t$, the number of
patients who can be helped, admits a deterministic FPT algorithm running in time
$\mathcal{O}^\star\!\bigl((4e)^t\bigr) \approx \mathcal{O}^\star\!\bigl(10.88^t\bigr)$.
\end{restatable*}

We next examine the prospects for polynomial kernels for the same parameter $t$. It is worth noting that exploring the existence of a polynomial kernel for a practically significant problem, such as kidney exchange, is highly motivating, as kernelization algorithms can substantially reduce the size of problem instances commonly encountered in practice.
While it is relatively straightforward to design FPT algorithms
for our problem without aiming for a specific running time, obtaining a polynomial kernel turns out to be unlikely under the standard complexity theoretic assumption $\mathrm{NP}\subseteq \mathrm{coNP}/\mathrm{poly}$. In particular, by a straightforward reduction from the \textsc{Directed $k$-Path} problem of Bodlaender, Jansen, and Kratsch~\cite{BodlaenderJK13}—which is known not to admit a polynomial kernel or compression unless $\mathrm{NP}\subseteq \mathrm{coNP}/\mathrm{poly}$—we show that such a compression remains impossible even for the larger parameter of $t+\ell_p+\ell_c+|\BB|$. This provides a nearly complete picture with respect to the
natural solution-size parameter.

\begin{restatable*}{thm}{nopolykernel}
\label{cor:kidney-no-poly}
Under the assumption $\mathrm{NP}\not\subseteq\mathrm{coNP}/\mathrm{poly}$, \KE admits no polynomial kernel (nor polynomial compression) parameterized by $t+\ell_p+\ell_c+|\mathcal{B}|$ where \BB is the set of altruistic vertices.
\end{restatable*}

We further strengthen the hardness landscape established by Johnson et al.~\cite{ijcai2024p9}, who proved that the problem remains \textsf{W[1]}-hard when parameterized by the treewidth of the underlying graph. Our results show that the problem remains \textsf{W[1]}-hard even when parameterized by the strictly larger parameter of pathwidth, and that even the addition of $\ell_c$ or $\ell_p$ alone does not make the problem (fixed-parameter)tractable. Specifically, we provide two separate hardness results for \KE problem: one on bounded-pathwidth DAGs (where $\ell_c = 0$) and another on bounded-pathwidth graphs with no altruistic vertices (where $\ell_p = 0$). These results follow from a parameter-preserving reduction from the \BP problem with unary-encoded weights~\cite{DBLP:journals/jcss/JansenKMS13}.

It is worth noting that when parameterized by all three parameters together, the problem becomes \FPT, as a corollary of the result by Maiti et al.~\cite{MaitiD22}, who proved fixed-parameter tractability when parameterized by \emph{treewidth} and $\max{\ell_p, \ell_c}$. Together, these results not only provide a more precise boundary of tractability but also give a complete picture of the parameterized complexity of \KE\ with respect to \emph{pathwidth}, $\ell_p$, and $\ell_c$.


\begin{restatable*}{thm}{whardnocycle} \label{thm:pathwidth-path}
\textsc{Kidney Exchange} is \textsf{W[1]}-hard when parameterized by the undirected pathwidth of the input graph, even when the graph is acyclic.
\end{restatable*}


\begin{restatable*}{thm}{noaltruisticvertex}\label{thm:pathwidth-cycle}
\textsc{Kidney Exchange} is \textsf{W[1]}-hard when parameterized by the undirected pathwidth of the input graph, even when no altruistic vertex is present. 
\end{restatable*}


While \cref{thm:pathwidth-path} already implies W-hardness when parameterized by the size of a Directed Feedback Vertex Set (DFVS), we strengthen this result by establishing para-NP-hardness even for the larger parameter $\text{DFVS} + \ell_p+\ell_c$, by showing that the problem remains NP-hard on DAGs even when $\ell_p$ is bounded. Note that since the associated graph is a DAG one may assume $\ell_c=0$ without loss of generality. Our proof uses a carefully designed two-step polynomial reduction from the \textsc{3-Partition} problem, where each set of a given weight corresponds to an equivalent altruistic path of proportional length in the resulting instance of the kidney exchange problem.


\begin{restatable*}{thm}{paranphardness}\label{thm:para-np-hard}
The \textsc{Kidney Exchange} problem on DAGs is NP-hard even when restricted to instances with a fixed constant path-length bound $\ell_p$.
\end{restatable*}

\begin{restatable*}{cor}{dfvsparanphard}\label{cor-para}
The \textsc{Kidney Exchange} problem is para-NP-hard with respect to DFVS+$\ell_p$+$\ell_c$.
\end{restatable*}

\begin{table*}[h]
	\begin{center}
	\begin{tabular}{|cc||cc|}\hline
		\BB & set of altruistic vertices & $t$ & target number of patients to receive kidneys\\
		$\el_p$ & length of the longest path allowed & $\el_c$ & length of the longest cycle allowed\\
		$\tw$ & treewidth of underlying undirected graph & $\theta$ & number of vertex types\\
		$\Delta$ & maximum degree of underlying undirected graph & $\el$ & $\max\{\el_p,\el_c\}$\\
		\hline
	\end{tabular}
\label{table1}
\caption{Notation table.}
\end{center}
\end{table*}

\subsection{Related Work}

Rapaport first proposed the idea of kidney exchange~\cite{rapaport1986case}. By 2005, over 60,000 patients were listed on the UNOS transplant registry~\cite{segev2005kidney}, and the New England Renal Transplant Oversight Committee approved a regional clearinghouse for exchanges in 2004~\cite{ke05}. Roth et al.~\cite{roth05} introduced strategy-proof, constrained-efficient mechanisms under 0–1 preferences. Some studies limit exchanges to cycles~\cite{constantino2013new,klimentova2014,sonmez2014altruistically}, while others allow both cycles and chains~\cite{manlove2015paired,glorie2014kidney,xiao2018exact}. All variants can be framed as graph packing problems, with connections to barter markets and set packing explored by Jia et al.~\cite{jia2017efficient}.

The problem is NP-hard, even with short cycles or chains~\cite{krivelevich2007approximation,AbrahamBS07,DBLP:journals/algorithmica/BelmonteHKKKKLO22}, and is para-NP-hard when parameterized by cycle and path lengths. Approximation algorithms have been proposed via set packing~\cite{krivelevich2007approximation,jia2017efficient}.

Practical methods include integer programming~\cite{manlove2015paired}, used by the UK's NHS for national donor matching. Dickerson et al.~\cite{dickerson2016position} emphasized the need to bound chain lengths for feasibility. Glorie et al.~\cite{glorie2014kidney} developed a branch-and-price algorithm supporting multi-objective optimization. Li et al.~\cite{li2014egalitarian} introduced a polynomial-time algorithm for Lorenz-dominant fractional matchings-fair, incentive-compatible, and implementable via lotteries. Approximation methods for short cycles were given by Biro et al.~\cite{biro2009maximum}, while Riascos-Alvarez et al.~\cite{klimentova2014} proposed a decomposition method for large instances, allowing prioritization between chains and cycles.

Dickerson et al.~\cite{dickerson2016position} also introduced the vertex type parameter to capture structural similarity in real graphs. Xiao and Wang~\cite{xiao2018exact} built on this with an exact algorithm of complexity $\OO\left(2^n n^3\right)$ and an FPT algorithm parameterized by vertex types, assuming unbounded chain and cycle lengths. Maiti and Dey~\cite{MaitiD22} initiated the study of parameterized complexity of the kidney exchange problem and presented a $\OO^\star\left(2^{\OO(t)}\right)$ time algorithm where $t$ is the maximum number of patients that receive a kidney; they also showed some other algorithmic and hardness results. Hebert-Johnson et al.~\cite{Hebert-JohnsonL24} improved this running time and presented a randomized algorithm for this problem running in time $\OO^\star\left(4^{t}\right)$ and some hardness results.

Lin et al.~\cite{lin2019randomized} studied the cycle-only case and proposed a randomized parameterized algorithm based on the number of recipients and maximum cycle length.

\section{Preliminaries}\label{sec:prelim}

For an integer $k$, we denote the sets $\{0,1,\ldots,k\}$ and $\{1,2,\ldots,k\}$ by $[k]_0$ and $[k]$ respectively.

A kidney exchange problem is formally represented by a directed graph $\GG=(\VV,\AA)$ which is known as the {\em compatibility graph}. A subset $\BB\subseteq\VV$ of vertices denotes {\em altruistic donors} (also called {\em non-directed donors}); the other set $\VV\setminus\BB$ of vertices denote a patient-donor pair who wish to participate in the kidney exchange program. We have a directed edge $(u,v)\in\AA$ if the donor of the vertex $u\in\VV$ has a kidney compatible with the patient of the vertex $v\in\VV\setminus\BB$. Kidney exchange happens either (i) along a {\em trading-cycle} $u_1, u_2, \ldots, u_k$ where the patient of the vertex $u_i\in\VV\setminus\BB$ receives a kidney from the donor of the vertex $u_{i-1}\in\VV\setminus\BB$ for every $2\leq i\leq k$ and the patient of the vertex $u_1$ receives the kidney from the donor of the vertex $u_k$, or (ii) along a {\em trading-chain} $u_1, u_2, \ldots, u_k$ where $u_1\in\BB, u_i\in\VV\setminus\BB$ for $2\le i\le k$ and the patient of the vertex $u_j$ receives a kidney from the donor of the vertex $u_{j-1}$ for $2\le j\le k$. Due to operational reasons, all the kidney transplants along a trading-cycle or a trading-chain should be performed simultaneously. This puts an upper bound on the length \el of feasible trading-cycles and trading-chains. We define the length of a path or cycle as the number of edges in it. The kidney exchange clearing problem is to find a collection of feasible trading-cycles and trading-chains which maximizes the number of patients who receive a kidney. Formally it is defined as follows.

\begin{definition}[\KE]\label{def:prob-ke}
Given a directed graph $\GG=(\VV,\AA)$ with no self-loops, an altruistic vertex set $\BB\subset\VV$, two integers $\el_p$ and $\el_c$ denoting the maximum length of respectively paths and cycles allowed, and a target $t$, compute if there exists a collection \CC of disjoint cycles of length at most $\el_c$ and paths with starting from altruistic vertices only each of length at most $\el_p$ which cover at least $t$ non-altruistic vertices. We denote an arbitrary instance of \KE by $(\GG,\BB,\el_p,\el_c,t)$.
\end{definition}

\subsection{Graph Theoretic Terminologies} In a graph \GG, $\VV[\GG]$ denotes the set of vertices in \GG and $\EE[\GG]$ denotes the set of edges in \GG. $\GG[\VV']$ denotes the induced subgraph on $\VV'$ where $\VV'\subseteq\VV[\GG]$. Two vertices $u$ and $v$ in a directed graph \GG are called vertices of the same type if they have the same set of in-neighbors and the same set of out-neighbors. \textcolor{black}{If there are no self loops in \GG, vertices of the same type form an independent set.} Treewidth measures how treelike an undirected graph is. We refer to \cite{CyganFKLMPS15} for an elaborate description of treewidth, tree decomposition, and nice tree decomposition. Since our graph is directed, whenever we mention treewidth of our graph, we refer to the treewidth of the underlying undirected graph; two vertices $u$ and $v$ are neighbors of the underlying undirected graph if and only if either there is an edge from $u$ to $v$ or from $v$ to $u$. \textcolor{black}{Also refer to the Table 1 for the important notations.}

\subsection{Parameterized Complexity}\label{subsec:pc}

A tuple $(x, k)$, where k is the parameter, is an instance of a parameterized problem. \emph{Fixed parameter tractability} (FPT) refers to solvability in time  $f(k) \cdot p(|x|)$ for a given instance $(x, k)$, where  $p$ is a polynomial in the input size $|x|$ and $f$ is an arbitrary computable function of $k$ . We use the notation $\OO^*(f(k))$ to denote $O(f(k)poly(|x|))$.

We say a parameterized problem is \PNPH if it is \NPH even for some constant values of the parameter.

\begin{definition}[Kernelization]~\cite{CyganFKLMPS15}
	A kernelization algorithm for a parameterized problem   $\Pi\subseteq \Gamma^{*}\times \mathbb{N}$ is an 
	algorithm that, given $(x,k)\in \Gamma^{*}\times \mathbb{N} $, outputs, in time polynomial in $|x|+k$, a pair 
	$(x',k')\in \Gamma^{*}\times
	\mathbb{N}$ such that (a) $(x,k)\in \Pi$ if and only if
	$(x',k')\in \Pi$ and (b) $|x'|,k'\leq g(k)$, where $g$ is some
	computable function.  The output instance $x'$ is called the
	kernel, and the function $g$ is referred to as the size of the
	kernel. If $g(k)=k^{O(1)}$, then we say that
	$\Pi$ admits a polynomial kernel.
\end{definition}

It is well documented that the presence of a polynomial kernel for certain parameterized problems implies that the polynomial hierarchy collapses to the third level (or, more accurately, \caveat{}). As a result, polynomial-sized kernels are unlikely to be present in these problems. We use a polynomial-parameter-preserving  reduction to demonstrate kernel lower bounds.


\section{A Faster Deterministic \FPT parameterized by $t$}
In this section, we design the fastest known deterministic FPT algorithm for the problem parameterized by $t$, consisting of three main steps. In the first step, we apply a color-coding procedure that ensures every vertex involved in the solution receives a distinct color which plays a crucial role in the second step, where we design a dynamic programming (DP) algorithm that systematically checks whether a solution exists in which exactly one vertex from each color subset is selected. During the gradual construction of entries corresponding to larger color sets, we use as a subroutine, a slight modified version of the standard DP algorithm ~\cite{CyganFKLMPS15} that runs in $2^{|S|}$ time to determine whether a cycle or altruistic path exists that uses all colors in $S$ exactly once. The formal algorithm is presented below.




\subsection{Color Coding}

We employ the color-coding method of Alon, Yuster, and Zwick~\cite{AlonYZ1995}. In this technique, each vertex of $G$ is assigned one of $t$ colors uniformly at random. With probability at least $e^{-t}$, the $t$ vertices forming the desired solution all receive distinct colors. This randomization coloring can be eliminated (made deterministic) by using families of $t$-perfect hash functions of size $e^t \cdot t^{\mathcal{O}(\log t)}$~\cite{NaorSchulmanSrinivasan95}. Henceforth, we assume access to such an $(n,t)$-perfect hash family. By iterating over all colorings in the family, we may assume without loss of generality that in some coloring, the vertices of the solution are colored distinctly. For every coloring in the family we employ the following DP algorithm that checks if a colorful solution is possible for the \KE instance:

\subsection{Dynamic Programming Formulation}

Let $[t]=\{1,2,\ldots,t\}$ denote the set of colors. For every subset $T' \subseteq [t]$, define


\[
DP(T') = 
\begin{cases}
    1 & \parbox[t]{0.6\columnwidth}{there exists a collection of vertex-disjoint cycles and altruistic paths
that uses exactly one vertex from each color in $T'$.}\\
&\\
0 & \text{otherwise}
\end{cases}
\]

The recurrence is as follows:
\[
DP(T') = 1 \iff \exists T'' \subset T' \text{ such that } DP(T'')=1 \text{ and either}
\]
\begin{itemize}
    \item there is a directed cycle with exactly one vertex from each color in $T'\setminus T''$, or
    \item there is an altruistic path with exactly one vertex from each color in $T'\setminus T''$.
\end{itemize}

For finding such cycles/altruistic paths in time $2^{|T'\setminus T''|}$, we use the subsequent step.

\subsection{Detecting Cycles and Altruistic Paths of Prescribed Colors}

The task of checking whether there exists a cycle or altruistic path using exactly a set of colors $S \subseteq [t]$ can be solved via a slight modification of the standard color-coding dynamic programming algorithm for detecting colorful paths~\cite{CyganFKLMPS15}. The standard DP algorithm runs in $2^{|S|}$ time and determines whether a directed path exists that uses all colors exactly once. For an altruistic path, the DP formulation is restricted so that the path must start at an altruistic vertex, while the rest of the algorithm remains unchanged. For a colorful cycle, one can guess an arc and then, in the DP algorithm, search for a path that uses all colors exactly once and starts and ends at the endpoints of the guessed arc. Thus these algorithms run in time $\mathcal{O}^\star(2^{|S|})$ and can be used to determine required colorful altruistic paths/cycles .



\subsection{Run Time Analysis}
Given a colored graph (with $t$ colors), the DP algorithm computes all DP entries for all subsets $T'\subseteq [t]$. For each $T'$, we consider all partitions $T'' \subseteq T'$ where we look for a colorful solution that must use all colors in $T''$ and check for an additional colorful cycle/altruistic path using all colors in $T'\setminus T''$ (which takes time $2^{|T'\setminus T''|}$). Hence, the total running time is therefore bounded by

$$\sum_{T' \subseteq [t],} \;\sum_{T'' \subseteq T'} 2^{|T'\setminus T''|} \cdot n^{\mathcal{O}(1)}$$

We now simplify the double sum. For each $i=|T'|$, there are $\binom{t}{i}$ choices of $T'$, and for each such $T'$:
\[
\sum_{T'' \subseteq T'} 2^{|T'\setminus T''|} = \sum_{j=0}^i \binom{i}{j} \cdot 2^{i-j} = 3^i.
\]
Thus the total runtime to compute the DP table is
\[
\sum_{i=0}^t \binom{t}{i} \cdot 3^i = (1+3)^t = 4^t.
\]
Hence provided a colored input graph, the algorithm checks if the graph has a solution of size $t$ in time $\mathcal{O}^\star(4^t)$ time. Using families of $t$-perfect hash functions of size $e^t \cdot t^{\mathcal{O}(\log t)}$~\cite{NaorSchulmanSrinivasan95} (these many instances of colorful graphs) for coloring, we can bound the overall running time of our algorithm by
\[
\mathcal{O}^\star\!\left( e^t \cdot t^{\mathcal{O}(\log t)} \cdot 4^t  \right).
\]
\



\fptalgo

\section{No poly kernel parameterized by $t+\ell_p+\ell_c (+|\mathcal{B}|)$}
In this section, we provide a polynomial-time, parameter-preserving reduction from \textsc{Directed $k$-Path}, which asks whether a directed path of length $k$ exists and is known not to admit a polynomial compression (see \cref{prop:dir-kpath-no-poly} below) under standard complexity theoretic assumptions, to an equivalent instance of \KE, thereby establishing the desired hardness result.

\begin{proposition}\cite{BodlaenderJansenKratsch14}
\label{prop:dir-kpath-no-poly}
{\sc Directed $k$-Path} (Long Directed Path parameterized by $k$) admits neither a polynomial kernel nor a polynomial compression unless
\(\mathrm{NP}\subseteq\mathrm{coNP}/\mathrm{poly}\). 
\end{proposition}

\noindent\textbf{Reduction:} Given an input instance $(G=(V,E),k)$ of \textsc{Directed $k$-Path}, we create an equivalent instance $(\mathcal{G}=(V\cup\{a\}, E\cup\{(a,v)\mid v\in V\}), \mathcal{B}=\{a\}, \ell_p = k+1, \ell_c = 0, t = k+1)$ of \textsc{Kidney Exchange}, i.e., we add an additional single altruistic vertex $a$ with arcs from $a$ to every vertex $v \in V$. We set $\ell_c = 0$, $\ell_p = k+1$, and require $t = k+1$ transplants. 

\begin{lemma}
$(G=(V,E),k)$ of \textsc{Directed $k$-Path} is a yes instance if and only if $(\mathcal{G}=(V\cup\{a\}, E\cup\{(a,v)\mid v\in V\}), \mathcal{B}=\{a\}, \ell_p = k+1, \ell_c = 0, t = k+1)$ is a yes instance of \KE.
\end{lemma}

\begin{proof}
In the forward direction, 
suppose $(G,k)$ is a yes instance and  contains a directed simple path $P=v_1, v_2,\cdots, v_k$ on $k$ vertices. Then in $\mathcal{G}$, $P'=a,v_1, v_2,\cdots, v_k$ is an altruistic path of length $k+1$ making $(\mathcal{G},\mathcal{B},\ell_p,\ell_c,t)$ a yes instance of \KE.

Conversely, suppose the constructed \textsc{Kidney Exchange} instance is a yes instance. Since $\ell_c=0$, no cycles are allowed, and because there is exactly one altruistic vertex $a$, the solution must contain a single altruistic path starting at $a$ of length $k+1$, necessarily of the form $P'=a,v_1, v_2,\cdots, v_k$, where the $v_i$ are distinct vertices in $V$. Removing the initial vertex $a$ yields the directed simple path $P=v_1, v_2,\cdots, v_k$ in $G$ making $(G,k)$ a yes instance.
\end{proof}
Since $k \in \mathcal{O}^\star(t+\ell_p+\ell_c+|\mathcal{B}|)$, we immediately get the following as a corollary.  

\nopolykernel


\section{\WOH~ness parameterized by Pathwidth}

In this section, we present our claimed hardness results, showing that \KE\ remains \WOH\ when parameterized by either pathwidth+$\ell_c$ or pathwidth+$\ell_p$. Specifically, we provide these two hardness results for \KE separately: one for bounded-pathwidth DAGs (where $\ell_c = 0$) and another for bounded-pathwidth graphs with no altruistic vertices (where $\ell_p = 0$). These results follow from the known \WOH of the \BP\ problem with unary-encoded weights~\cite{DBLP:journals/jcss/JansenKMS13}. 
Note that, Maiti et al.~\cite{MaitiD22} showed that \KE\ is \FPT\ when parameterized by the combined parameter of \emph{treewidth}, $\ell_p$, and $\ell_c$, and consequently also \FPT\ when parameterized by the combined parameter of \emph{pathwidth}, $\ell_p$, and $\ell_c$. 
Together with our result, this provides a complete characterization of the parameterized complexity of \KE\ with respect to \emph{pathwidth}, $\ell_p$, and $\ell_c$.
We begin with the formal definition of the \BP\ problem, followed by our claimed parameter-preserving reductions to \KE.

\begin{definition}[\BP]
    Given a set \UU of $n$ items with positive integral sizes $w_1,\ldots,w_n$ and $k$ bins each of size $\frac{W}{k}$ where $W=\sum_{i=1}^n w_i$, compute if there exists a partition of the items into $k$ parts such that the total size of every part is exactly $\frac{W}{k}$. We denote an arbitrary instance of \BP by $(\UU,k)$.
\end{definition}

We know that \BP is \WOH parameterized by $k$ even when all the numbers in the input are encoded in unary~\cite{DBLP:journals/jcss/JansenKMS13}. The high-level idea of our reduction is as follows. Each item from the bin packing becomes a small chain of donors and recipients, and each bin becomes a longer chain that starts from an altruistic donor.
We connect these parts in such a way that deciding which item chains attach to which altruistic chain is exactly the same as deciding which items go into which bins.
If the items can be evenly packed into bins, then the corresponding kidney exchange instance also has a perfect matching of transplants, and vice versa.
Because our construction maintains the acyclicity of the graph and does not significantly increase its structural complexity, this demonstrates that the kidney exchange problem remains challenging even for graphs with low pathwidth. We now present the formal proof.

\whardnocycle


\begin{proof}
    We reduce from \BP parameterized by $k$. Let $(\UU=\{w_i: i\in[n]\},k)$ be an arbitrary instance of \BP. We assume without loss of generality that $w_i\ge 3k^2n^2$ for every $i\in[n]$; if it is not the case, then we multiply $w_i$ with $3k^2n^2$ for every $i\in[n]$. We now construct an instance $(\GG,\BB,\el_p,\el_c,t)$ of \KE. For every $w_i\in\UU$, we have a directed path $p_i$ starting from $u_{i,s}$ to $u_{i,\lambda}$ consisting of $w_i$ vertices. We also have a directed path $q_j=(a_{j,1},\ldots,a_{j,n})$ for every $j\in[k]$. Other than the edges in the above paths, we also have the following edges. For every $i\in[n]$, we have an edge from $a_{j,i}$ to $u_{i,s}$. We also have an edge from $u_{i,\lambda}$ to $a_{j,i+1}$ for every $i\in[n-1]$. This finishes the description of \GG. We note that \GG is an acyclic graph with a following topological order.
    \[a_{1,1},\ldots,a_{k,1},\vec{\VV}(p_1),a_{1,2},\ldots,a_{k,2},\vec{\VV}(p_2),\ldots,\vec{\VV}(p_n)\]
    where $\vec{\VV}(p_i)$ is the vertices in $p_i$ ordered as the directed path $p_i$ from $u_{i,s}$ to $u_{i,\lambda}$ for $i\in[n]$. The set \BB of altruistic vertices is $\{a_{j,1}:j\in[k]\}$. We set $\el_p=\frac{W}{k}+n$ where $W=\sum_{i=1}^n w_i$ and $\el_c$ arbitrarily --- value of $\el_c$ does not matter since the graph of the \KE instance is acyclic. We define $t=|\VV(\GG)\setminus\BB|$, that is, the target is to cover every non-altruistic vertices in \GG. We now claim that the two instances are equivalent.

    In one direction, let us assume that the \BP instance is a \yes instance. Let $(X_1,\ldots,X_k)$ be a partition of \UU such that the total size of items in $X_j$ is $\frac{W}{k}$ for every $j\in[k]$. For every $j\in[k]$, we define a path $\PP_j$ as
    \[\PP_j = a_{j,1},Q_1,a_{j,2},Q_2,\ldots,a_{j,n},Q_n\]
    where $Q_i$ is $\vec{\VV}(p_i)$ if $w_i\in X_j$ and the empty set otherwise. We first observe that, since $(X_1,\ldots,X_k)$ is a partition of \UU, $(\VV(\PP_1),\ldots,\VV(\PP_k))$ is a partition of $\VV[\GG]$. Also, since the total size of items in $X_j$ is $\frac{W}{k}$, each $\PP_j$ contains $\frac{W}{k}+n$ vertices which is $\el_p$. Hence, the \KE instance is a \yes instance.

    In the other direction, let us assume that the \KE instance is a \yes instance. Since $\el_p=\frac{W}{k}+n, t=|\VV(\GG)\setminus\BB|,$ and $|\BB|=k$, there exist $k$ disjoint paths $\PP_1,\ldots,\PP_k$ in \GG each of length $\el_p=\frac{W}{k}+n$ covering all the vertices of \GG. For every $i\in[n]$, we observe that each path $\PP_j, j\in[k]$ must either contain all the vertices of the path $p_i$ or none of its vertices, since $\PP_j$ can neither end inside $p_i$ as then the next vertex cannot be covered by any other path nor move from inside $p_i$ to outside $p_i$ as there is no such edge in \GG. We now define $X_j=\{w_i: i\in[n], \PP_j \text{ contains the path }p_i\}$ for $j\in[k]$. We claim that $\sum_{x\in X_j}x=\frac{W}{k}$. Indeed, otherwise $\PP_j$ covers either at most $\frac{W}{k}+n^2-3k^2n^2$ vertices or at least $\frac{W}{k}+n^2+3k^2n^2$ vertices since $w_i\ge 3k^2n^2$ for every $i\in[n]$. We now refute both possibilities, thereby proving the claim. The path $\PP_j$ cannot cover at least $\frac{W}{k}+n^2+3k^2n^2$ vertices since $\el_p=\frac{W}{k}+n$. Also, the path $\PP_j$ cannot cover at most $\frac{W}{k}+n^2-3k^2n^2$ vertices, since then the remaining $k-1$ paths each of length at most $\el_p$ cannot cover all the non-altruistic vertices in \GG. Hence, the \KE instance is a \yes instance too.

    Finally, we observe that the pathwidth of \GG is at most $k-1$ since the following sequence of bags each containing at most $k$ vertices is a valid path decomposition of \GG.
    \[\{a_{1,1},\ldots,a_{k,1}\},\vec{\VV}(p_1),\{a_{1,2},\ldots,a_{k,2}\},\vec{\VV}(p_2),\ldots,\vec{\VV}(p_n)\]
    Now the result follows since \BP is \WOH parameterized by $k$.
\end{proof}

We can modify the reduced instance of \Cref{thm:pathwidth-path} to show that \KE is \WOH with respect to pathwidth even when there is no altruistic vertex. We add an edge from the last vertex of the path $p_n$ to each altruistic vertex. Everything else remains the same. We present the detailed proof below.

\noaltruisticvertex


\begin{proof}
    From an arbitrary instance of \BP, we construct a \KE instance exactly as in the proof of \Cref{thm:pathwidth-path} except we also add an edge from $u_{n,\lambda}$ (which is the last vertex of the path $p_n$) to $a_{j,1}$ for every $j\in[k]$, there is no altruistic vertex, $\el_c=\frac{W}{k}+n$, $\el_p$ is anything (since there is no altruistic vertex, the value of $\el_p$ is irrelevant), and $t=\VV[\GG]$. We now prove equivalence the \KE instance with the \BP instance.

    In one direction, let us assume that the \BP instance is a \yes instance. Let $(X_1,\ldots,X_k)$ be a partition of \UU such that the total size of items in $X_j$ is $\frac{W}{k}$ for every $j\in[k]$. For every $j\in[k]$, we define a cycle $\CC_j$ as
    \[\CC_j = a_{j,1},Q_1,a_{j,2},Q_2,\ldots,a_{j,n},Q_n,a_{j,1}\]
    where $Q_i$ is $\vec{\VV}(p_i)$ if $w_i\in X_j$ and the empty set otherwise. We first observe that, since $(X_1,\ldots,X_k)$ is a partition of \UU, $(\VV(\CC_1),\ldots,\VV(\CC_k))$ is a partition of $\VV[\GG]$. Also, since the total size of items in $X_j$ is $\frac{W}{k}$, each $\CC_j$ contains $\frac{W}{k}+n$ vertices which is $\el_c$. Hence, the \KE instance is a \yes instance.

    In the other direction, let us assume that the \KE instance is a \yes instance. Since $\el_c=\frac{W}{k}+n, t=|\VV(\GG)\setminus\BB|,$ and $|\BB|=k$, there exist $k$ disjoint cycles $\CC_1,\ldots,\CC_k$ in \GG each of length $\el_c=\frac{W}{k}+n$. For every $i\in[n]$, we observe that each cycle $\CC_j, j\in[k]$ must either contain all the vertices of the path $p_i$ or none of its vertices, since $\CC_j$ can neither end inside $p_i$ as then the next vertex cannot be covered by any other path nor move from inside $p_i$ to outside $p_i$ as there is no such edge in \GG. We now define $X_j=\{w_i: i\in[n], \CC_j \text{ contains the path }p_i\}$ for $j\in[k]$. We claim that $\sum_{x\in X_j}x=\frac{W}{k}$. Indeed, otherwise $\PP_j$ covers either at most $\frac{W}{k}+n^2-3k^2n^2$ vertices or at least $\frac{W}{k}+n^2+3k^2n^2$ vertices since $w_i\ge 3k^2n^2$ for every $i\in[n]$. We now refute both the possibilities, thereby proving the claim. The cycle $\CC_j$ cannot cover at least $\frac{W}{k}+n^2+3k^2n^2$ vertices since $\el_c=\frac{W}{k}+n$. Also, the cycle $\CC_j$ cannot cover at most $\frac{W}{k}+n^2-3k^2n^2$ vertices, since then the remaining $k-1$ cycles each of length at most $\el_c$ cannot cover all the non-altruistic vertices in \GG. Hence, the \KE instance is a \yes instance too.

    Finally, we observe that the pathwidth of \GG is at most $2k-1$ since the following sequence of bags each containing at most $k$ vertices is a valid path decomposition of \GG where $\AA=\{a_{1,1},\ldots,a_{k,1}\}$ and $\AA\cup\vec{\VV}(p_i):=\AA\cup\{u_{i,s}\},\ldots,\AA\cup\{u_{i,\lambda}\}$ for $i\in[n]$.
    

    \begin{multline*}
\{a_{1,1}, \ldots, a_{k,1}\},\; 
\AA \cup \vec{\VV}(p_1),\;
\AA \cup \{a_{1,2}, \ldots, a_{k,2}\},\;
\AA \cup \vec{\VV}(p_2),\\ \ldots,
\AA \cup \vec{\VV}(p_n)
\end{multline*}
    Now the result follows since \BP is \WOH parameterized by $k$.
\end{proof}


\section{NP-hardness on DAGs}

To establish \textsf{NP}-hardness in our setting, we first reduce a classical \textsc{3-Partition} instance into an equivalent instance of \textsc{Fixed-Size-3-Partition}, where each feasible group is required to contain exactly three elements that collectively achieve a specified \emph{target sum}. This reduction is crucial for our subsequent reduction to \KE, in which altruistic paths of fixed lengths in the constructed solution correspond precisely to the groups attaining the \emph{target sum}.


\subsection{Reduction from \textsc{3-Partition} to \textsc{Fixed-Size-3-Partition}}

In the standard \textsc{3-Partition} problem, we are given a multiset 
$A = \{a_1,\dots,a_{3m}\}$
of positive integers and a target value $B$. The task is to decide whether $A$ can be split into exactly $m$ disjoint triples whose elements each sum to $B$. This problem remains \textsf{NP}-hard even under strong assumptions such as $B/4 < a_i < B/2$ and for constant values of $B$. For our reduction, we do not rely on any such promise; instead, we enforce triple cardinality groups via an additive shift.

\begin{definition}(\textsc{Fixed-Size-3-Partition})
  In this problem, we are given a multiset 
$A = \{a_1, \dots, a_{3m}\}$ of positive integers and a target value $B$, 
with the additional property that any subset of elements summing to $B$ contains exactly three elements. 
The question is whether there exists a partition of $A$ into $m$ disjoint groups, each having a total sum equal to $B$.
\end{definition}

Below, we provide a polynomial-time reduction from \textsc{3-Partition} to \textsc{Fixed-Size-3-Partition}, thereby establishing the \textsf{NP}-hardness of the latter problem.
The key idea is to enforce that every feasible group contains exactly three elements by applying a uniform additive shift to all item values. By choosing a constant $C$ much larger than any of the original values, any subset of fewer than three elements cannot reach the new target sum (which is larger than $3C$ and smaller than $4C$), while any subset of more than three elements necessarily exceeds it. Thus, only groups of size three can potentially satisfy the target constraint, and among such groups, feasibility still depends on the specific magnitudes of the $a_i$. 
Once three elements are selected, subtracting the shift restores the original sum, ensuring that each valid group corresponds precisely to a triple in the original instance. The formal construction and subsequent arguments are presented below.

\medskip
\noindent\textbf{Construction.}
Given an instance $(A, B)$ of \textsc{3-Partition}, where 
$A = \{a_1, \dots, a_{3m}\}$, 
we define a constant 
$C := 10B$ 
and construct an equivalent instance $(A=\{a'_1,a'_2,\dots a'_{3m}\},B'=3C+B)$ of \textsc{Fixed-Size-3-Partition} where
$a'_i := a_i + C \quad \text{for } i \in [3m]$.

\begin{lemma}
$(A,B)$ is a yes-instance of \textsc{3-Partition} if and only if $(A',B')$ is a yes-instance of \textsc{Fixed-Size-3-Partition}.
\end{lemma}

\begin{proof}
Before presenting the proof, we note that since all $a_i$ are positive integers, the input instance of \textsc{3-Partition} cannot contain any element $a_i > B$. 
With this observation, consider any group of size at most 2 in the constructed instance of \textsc{Fixed-Size-3-Partition}, can have a maximum sum of $2B+2C<3C<B'$. On the other hand, any group containing four elements has a minimum possible sum of $4C > B'$, which further implies that $k_r \le 3$. 
Hence, the size of any group achieving the target sum $B'$ must be exactly $3$.

\textbf{(If).} Suppose $(A,B)$ is a yes-instance of \textsc{3-Partition} and $A$ can be partitioned into $m$ disjoint triples
$
\big\{\{a_{r_1}, a_{r_2}, a_{r_3}\}\big\}_{r \in [m]} 
$
with
$
\sum_{j=1}^3 a_{{r_j}} = B.
$
Then for each such triple the corresponding triples of $
\big\{\{a'_{r_1}, a'_{r_2}, a'_{r_3}\}\big\}_{r \in [m]} 
$ forms a valid partition of $A'$ into $m$ groups since
$
\sum_{j=1}^3 a'_{{r_j}}
= \sum_{j=1}^3 (a_{{r_j}} + C)
= B + 3C
= B',
$ making  $(A',B')$  a yes-instance of \textsc{Fixed-Size-3-Partition}.

\smallskip
\textbf{(Only if).} Conversely, suppose $(A',B')$ is a yes-instance of \textsc{Fixed-Size-3-Partition} and $A'$ admits a partition into $m$ groups each of size 3 as argued in the beginning part of our proof, $\big\{\{a'_{r_1}, a'_{r_2}, a'_{r_3}\}\big\}_{r \in [m]}$ 
each with sum $B' = 3C + B$. But then $\big\{\{a_{r_1}, a_{r_2}, a_{r_3}\}\big\}_{r \in [m]}$ forms a valid partition of $A$ into $m$ groups 
since
$
\sum_{j=1}^3 a_{{r_j}}
= \sum_{j=1}^3 (a'_{{r_j}} - C)
= B' - 3C
= B,
$ making  $(A,B)$ is a yes-instance of \textsc{3-Partition}.
\end{proof}

Since the magnitudes of all elements including the target sum $B'$ increase by at most a factor linear in $B$, the reduction is polynomial in the input size.  And \textsc{3-Partition} being  \textsf{NP}-hard, immediately implies the following.

\begin{theorem}
\textsc{Fixed-Size-3-Partition} is \textsf{NP}-hard even for constant values of $B'$.
\end{theorem}

\subsection{Reduction from \textsc{Fixed-Size-3-Partition} to \KE on DAGs}

We next prove that the \KE problem on directed acyclic graphs is \textsf{NP}-hard, even for a fixed(constant) value of $\ell_p$ by showing a polynomial reduction from \textsc{Fixed-Size-3-Partition} problem.


\noindent\textbf{Construction.} Let $(A' = \{a'_1,\dots,a'_{3m}\}, T')$ be an instance of \textsc{Fixed-Size-3-Partition}. We construct an equivalent instance $(\mathcal{G},\mathcal{B},\ell_p,\ell_c,t)$ of \KE as follows.

\begin{itemize}
\item For each $a'_i \in A'$, create a directed (element-)path $P_i$ of $a'_i$ many vertices (see Figure~\ref{fig:hardness_dag}).

\item Add $m$ altruistic vertices $\mathcal{B}=\{s_1,\dots,s_m\}$. From each $s_j$, $j\in [m]$, add an outgoing edge/arc to the first vertex of every $P_i$, $i\in[3m]$.

\item Index the paths $P_1,\dots,P_{3m}$ and, for every $i<j$, add an outgoing edge/arc (concatenation arcs) from the last vertex of $P_i$ to the first vertex of $P_j$. This preserves acyclicity and provides an opportunity for paths to concatenate.

\item Set $\ell_p := B' + 1$ and $t=|V(\mathcal{G})|=\sum_{i\in [3m]}a'_i+m$. We consider vertex-disjoint altruistic paths of length at most $\ell_p$ and are interested in solutions that cover all vertices of $\mathcal{G}$.
\end{itemize}

\begin{figure}
\arxivversion{\begin{center}\includegraphics{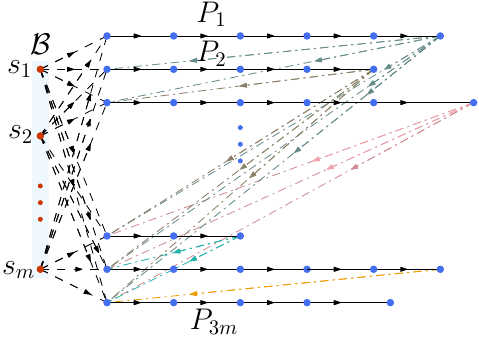}\end{center}}
\aamasversion{\includegraphics[width=2.2in]{hardness_dag.pdf}}
\caption{Reduction from \textsc{Fixed-Size-3-Partition} problem}
\label{fig:hardness_dag}
\end{figure}

\begin{lemma}
$(A',B')$ is a yes-instance of \textsc{Fixed-Size-3-Partition} if and only if $(\mathcal{G},\mathcal{B},\ell_p,\ell_c,t)$ is a yes instance of \KE.
\end{lemma}

\begin{proof}
\textbf{(If).} 
Suppose $(A',B')$ is a yes-instance of \textsc{Fixed-Size-3-Partition} where $A'$ can be partitioned into $m$ disjoint triples $\big\{\{a'_{r_1},a'_{r_2},a'_{r_3}\}\big\}_{r\in[m]}$ each summing to $T'$. For each triple $\{a'_{r_1},a'_{r_2},a'_{r_3}\}$ (where we assume w.l.o.g. that $r_1<r_2<r_3$), we construct an altruistic path in solution of length exactly $\ell_p$ that starts at $s_r$ followed by the paths $P_{r_1}$, $P_{r_2}$, and $P_{r_3}$. These $m$ altruistic paths together cover all vertices while each being of length exactly $\ell_p$.

\smallskip
\textbf{(Only if).} 
Conversely, suppose $(\mathcal{G}, \mathcal{B}, \ell_p, \ell_c, t)$ is a yes-instance of \KE, admitting a collection of $m$ vertex-disjoint altruistic paths that together cover all vertices of $\mathcal{G}$. 
Note that each such altruistic path must be of length exactly $\ell_p$ in order to cover all vertices of the graph. 
Moreover, each altruistic path begins at an altruistic vertex $s_r$ and immediately enters some element-path $P_i$ at its first vertex, since the only outgoing edges from altruistic vertices lead to first/initial vertices of the paths constructed from elements.  
By construction, each element-path $P_i$ is a directed path with only incoming edges to its first vertex, and only outgoing edges  from its last vertex. 
Consequently, an altruistic path can only traverse entire element-paths; it cannot enter or exit from an internal vertex of an element-path. Additionally, no altruistic path can end at any vertex other than the last vertex of an element-path, since neither the internal vertices nor the last vertex can be covered by any other altruistic path. In particular, this ensures that the element-paths are assigned integrally/completely to the altruistic paths. Finally, the concatenation arcs enforce that any altruistic path traverses element-paths in increasing index order. The integral/complete assignment of element paths to altruistic paths along with the fact that every altruistic path has length exactly $\ell_p=1+B'$ enforces that every altruistic path contains precisely 3 element paths. Therefore the altruistic paths induce a partition of the index set $\{1,\dots,3m\}$ into exactly $m$ groups (one group per altruistic path), where a group corresponding to an altruistic path consists of the indices of the element-paths it traverses in increasing order. This implies that the sum total of all non-altruistic vertices in any such group of size 3 is exactly equal to $B'$. Hence this partition of $3m$ indices also gives a partition of $A'$ into $m$ groups each summing up to $B'$, making $(A',B')$ a yes instance of \textsc{Fixed-Size-3-Partition}. 

\end{proof}

\smallskip
\noindent\textbf{Complexity.}  
The reduction introduces $\mathcal{O}(\sum_i a'_i)$ vertices and $\mathcal{O}((\sum_i a'_i)^2)$ edges, both polynomial in the input size. And, $\ell_p$ being $B'+1$ implies the following.

\paranphardness

\dfvsparanphard

\section{Conclusion}

In this work, we advanced the theoretical understanding of the Kidney Exchange problem by establishing both faster parameterized algorithms and tighter lower bounds. Specifically, we presented the fastest known deterministic fixed-parameter tractable (FPT) algorithm parameterized by the number of patients helped, improving the previous best running time from $\OO^\star\left(14^t\right)$ to $\OO^\star\left((4e)^t\right)\approx \OO^\star\left(10.88^t\right)$. Our approach combines color-coding, derandomized through $t$-perfect hash families, with a subset-based dynamic programming formulation that systematically constructs feasible cycles and altruistic chains. This yields a clean, purely combinatorial algorithm that is efficient, transparent, and asymptotically optimal among deterministic approaches of this kind.

Beyond algorithmic improvements, we showed that the problem admits no polynomial kernel when parameterized by $t+\el_p+\el_c+|\BB|$, unless $\mathrm{NP}\not\subseteq\mathrm{coNP}/\mathrm{poly}$, thus ruling out the possibility of efficient instance compression under standard complexity assumptions. We further strengthened the known hardness results by proving that Kidney Exchange is \WOH parameterized by pathwidth, even under strong structural restrictions such as acyclicity or the absence of altruistic donors. This settles a natural open question regarding the tractability of the problem on graphs of bounded pathwidth. Complementing these results, we also demonstrated para-NP-hardness with respect to the combined parameters DFVS,$\el_p, \el_c,$ thereby completing a near-exhaustive map of the problem’s parameterized complexity landscape.

Taken together, our findings delineate a sharp boundary between tractable and intractable regimes for Kidney Exchange under a variety of structural and quantitative parameters. From a broader perspective, these results strengthen the theoretical foundations of algorithmic kidney exchange, clarifying what forms of efficiency are achievable through parameterization and what barriers are provably inherent.

There are several promising avenues for future research. First, an immediate direction is to further improve the running time of FPT algorithms, either by tightening constants in the exponent or by exploring alternative algorithmic paradigms such as representative sets or inclusion–exclusion–based methods. Another direction is to investigate lower bounds under the Exponential Time Hypothesis (ETH) or Strong ETH, which could clarify whether our current algorithm is asymptotically optimal. It would also be valuable to reduce the polynomial dependence on the input size $n$, as this term often dominates in practical settings.

From a modeling perspective, our work focuses on the classical cycle-and-chain formulation. Extending these techniques to richer and more realistic models—including weighted or utility-based exchanges, dynamic or online formulations, and fairness- or priority-aware variants—remains an important open challenge. Such extensions would bridge the gap between parameterized theory and the design of scalable, value-aligned kidney exchange algorithms used in national registries. Finally, exploring kernelization-inspired preprocessing heuristics or hybrid parameterizations that combine structural and solution-size parameters may provide new insights into balancing theoretical rigor with real-world efficiency.

\section*{Acknowledgement}
Aritra Banik acknowledges support from the Anusandhan National Research Foundation (ANRF) (erstwhile Science, Education, and Research Board (SERB)), Government of India, via the project MTR/2022/000253.
Palash Dey thanks the Anusandhan National Research Foundation (ANRF) (erstwhile Science, Education, and Research Board (SERB)), Government of India, for supporting this work through Core Research Grant under file no. CRG/2022/003294.

\aamasversion{\bibliographystyle{ACM-Reference-Format}}
\arxivversion{\bibliographystyle{alpha}}
\bibliography{references}

\end{document}